\documentclass[12pt,a4paper]{article}
\usepackage{amssymb,amsmath,amsthm,enumerate}
\usepackage{graphicx}

\theoremstyle{plain}
\newtheorem{theorem}{Theorem}[section]

\newtheorem{proposition}[theorem]{Proposition}

\theoremstyle{definition}

\newtheorem{claim}[theorem]{Claim}

\newcommand\cref[1]{Corollary~\ref{cor:#1}}

\title{Between the deterministic and non-deterministic query complexity}
\author{D\'aniel Gerbner\\
\medskip 
\small Alfr\'ed R\'enyi Institute of Mathematics, Hungarian Academy of Sciences\\
\small P.O.B. 127, Budapest H-1364, Hungary.\\
\medskip
\small \texttt{gerbner@renyi.hu}
}

\begin{document}

\maketitle

\begin{abstract}
    We consider problems that can be solved by asking certain queries. The deterministic query complexity $D(P,n)$ of a problem $P$ is the smallest number of queries needed to ask in order to find the solution with an input of size $n$ (in the worst case), while the non-deterministic query complexity $D_0(P,n)$ is the smallest number of queries needed to ask, in case we know the solution, to prove that it is indeed the solution (in the worst case). Equivalently, $D(P,n)$ is the largest number of queries needed to find the solution in case an Adversary is answering the queries, while $D_0(P,n)$ is the largest number of queries needed to find the solution in case an Adversary chooses the input.
    
    We define a series of quantities between these two values, $D_k(P,n)$ is the largest number of queries needed to find the solution in case an Adversary chooses the input, and answers the queries, but he can change the input at most $k$ times. We give bounds on $D_k(P)$ for various problems $P$.
    
\end{abstract}

\section{Introduction}

In this paper, we consider problems that can be solved by asking certain queries. Somewhat loosely, a problem consists of a function and some kind of description of allowed queries. We are given a function $f:X\rightarrow Y$. $X$ is the set of inputs, and our goal is to determine $f(x)$ for an $x\in X$ that is unknown to us. We can ask queries of certain type in order to do that. We need to construct an algorithm, i.e. a strategy that describes what is the next query after a series of queries and answers. This can be represented by a \textit{decision tree}, a tree where the vertices are potential queries, the first query is the root, the possible answers to a query are represented by edges going out from the corresponding vertex, and at the other end of an edge we have the next query.   

The most typical goal is to find algorithms that minimize the worst case query complexity $D(P,n)$ of a problem $P$, i.e. the maximum number of queries that the algorithm needs to ask to solve the problem on an input of size $n$.

A useful approach is to consider the problem as a game between a player who asks the queries and tries to solve the problem as fast as possible (we will call him \textit{Questioner}), and another player, usually called \textit{Adversary}, who answers the queries and tries to postpone the solution. Rather than choosing an input element $x\in X$, the Adversary can answer arbitrarily. However, his answers have to be consistent with at least one input. The game finishes if Questioner can find the value of $f$, assuming that the answers are consistent with at least one input. 

Assume that we can find an algorithm for the Adversary such that no matter what strategy the Questioner has, at least $m$ queries have to be asked before the game finishes. Then it is easy to see that $D(P,n)\ge m$. Indeed, against any particular strategy of the Questioner, the Adversary's answers are consistent with an input $x$. That means the particular strategy needs to ask at least $m$ queries if $x$ happen to be the input, and this holds for every strategy, thus $D(P,n)\ge m$.

On the other hand, even the best strategy for the Adversary cannot force Questioner to ask more than $D(P,n)$ queries.

\bigskip

Above we were looking for the algorithm that minimized the number of queries for the worst input, i.e. the minimum (among algorithms) of the maximum (among inputs) of the queries needed to solve the problem for that input. One can look instead for the maximum (among inputs) of the minimum (among algorithms) of the queries needed. In other words, previously we evaluated algorithms based on the worst input, and picked the best algorithms, but now we are going to evaluate inputs based on the best algorithm, and then pick the worst input.

A more typical approach to this version is that we can guess the input $x$, but we still have to prove that the value of $f$ is indeed $f(x)$. Our knowledge on $x$ only helps us choose the queries that we ask. It is called the \textit{non-deterministic} version of the problem. In the decision tree, it means that instead of the longest path from a root to a leaf, we are looking for the shortest such path. Its length is called the \textit{non-deterministic query complexity} or \textit{certificate complexity} of the problem, and we denote it by $D_0(P,n)$. Note that non-deterministic complexity is usually defined in a non-symmetric way. In particular, if $f$ has only two possible values 0 or 1, it is possible that for every $x$ with $f(x)=1$, the proof of this fact is short, while for some $y$ with $f(y)=0$, the proof of this fact is long.

If we consider the problem as a game, the change is that Adversary has to pick the input $x$ immediately. Earlier we looked for the number of queries needed if both players play optimally. Now we still have the same assumption, Questioner plays optimally, and the Adversary has to choose the input that gives the largest number of queries. We can look at the optimal play of the Questioner, as if he knows (because he guessed correctly) the input.

\bigskip

One can look at Adversary strategies in the deterministic version the following way. Instead of saying that the Adversary does not have to pick the input, we can say that he has to pick one, but he can \emph{change} it. Of course, when the input is changed, the new one has to be consistent with all the earlier answers.

In this paper we are going to limit the Adversary's ability. We are going to assume that the input can be changed at most $k$ times. We remark that the earlier answers to the queries have to be satisfied, we say that those are \textit{fixed}. We call the set of steps between two changes of the input a \textit{round}.

We are going to look at the game from the Adversary's point of view. If we are interested in the best case for the Adversary (i.e. the worst case for the Questioner), then by the above arguments, $D(P,n)$ queries are needed to finish the game if he can change the input arbitrarily and $D_0(P,n)$ queries are needed to finish the game if he cannot change the input at all.

Hence we are interested in the worst case for the Adversary. He can choose the input and then change it at most $k$ times, how long can he extend the game? We denote this 
problem by $P(k)$ and the 
number of queries needed in this case (with an input of size $n$) by $D_k(P,n)$. We assume Questioner asks the best possible queries all the time. Therefore; we can assume the Questioner actually knows all the time, what the current input is, and his goal is not to find the value of $f$, but to prove that it is that value (i.e. to find a certificate). In particular, if $k=0$, then Adversary does not play after choosing $x$, and Questioner knows $x$. This is indeed the non-deterministic version, thus we do not denote two different things by $D_0(P,n)$. 

On the other hand, if $k\ge D(P,n)-1$, then the Adversary can change the input as much as he wants in the first $D(P,n)-1$ steps, and the best strategy of Questioner ends before any further chance for the Adversary to change the input. Thus we have $D_k(P,n)=D(P,n)$ if $k\ge D(P,n)-1$.

The above observations show that for any $k$, $D_k(P,n)$ is between the non-deterministic and the deterministic complexity of the problem $P$.  We continue with two other simple observations.

\begin{proposition}\label{altalanos}
(i) $D_k(P,n)\le D(P,n)$.

(ii) $D_k(P,n)\le \min\{\sum_{i=1}^l D_{j_i}(P,n): \, l,i_1,\dots,i_l  \text{ are such that } \sum_{i=1}^l (j_i+1)>k\}$.
\end{proposition}

\begin{proof}
$D(P,n)$ is an upper bound by definition. To show (ii), the Questioner runs the optimal algorithm for $P(j_1)$. If it does not solve the problem, the Adversary had to change the input more than $j_1$ times. Then the Questioner runs the optimal algorithm for $P(j_2)$, and so on till $P(j_l)$. If the problem is never solved, in each part the Adversary had to change the input at least $j_i+1$ times, thus altogether more than $k$ times, a contradiction.
\end{proof}

\begin{proposition}\label{also}
$D_k(P,n)\ge \min \{k+1,D(P,n)\}$.
\end{proposition}

\begin{proof} Let Strategy $A$ of the Adversary be one that forces at least $D(P,n)$ questions in case he can change the input as many times as he wants.
Let us assume first $k+1\ge D(P,n)$. Then the Adversary has the freedom to change the input in the first $D(P,n)-1$ steps, thus he can follow Strategy $A$. If $k+1<D(P,n)$, then the Adversary follows the same strategy. He can lose only when he cannot change the input anymore, thus no earlier than after the $k+1$st query.
\end{proof}

\subsection*{Motivation}

We have multiple different ways to look at the new question initiated here, and they lead to different motivations and different terminology.

Our main motivation is to understand the difference between deterministic and non-deterministic query complexity more, by examining what happens in the middle ground.

If we consider the problem as a game, we can look at the non-deterministic version, where the Adversary has to pick an input. If we look at the deterministic version afterwards, we can consider Adversary's strategy of changing the input as cheating. For example, in the well-known game of Battleship, a player can rearrange his ships without the opponent noticing. This is useful (in increasing the number of guesses needed to sink all the ships from $D_0(P,n)$ to $D(P,n)$), but it is obviously cheating. In this setting, we study the limitations of cheating. Note that in this setting the Questioner also cheats by peeking.


In a sense, we also study how robust our algorithms are. When we look at the optimal algorithms, not just the values of $D_k(P,n)$, we can see how much we have to change them if the input changes.

Let us mention that it is not unimaginable that we meet our problem when examining actual algorithms. It is not unusual that we have a strong assumption on the input. Yet a query can have an unexpected answer. That results in a paradigm shift, a complete reevaluation of the input, which leads to a new, equally strong assumption. Assuming there are at most $k$ paradigm shifts is very unnatural, but assuming there are not too many shifts makes sense.

Finally, we can look at our problems from a different angle. Instead of having a fixed input, we are given a dynamic, ever-changing system. When we ask a query, rather than getting an existing (but unknown to us) information about the input, we \textit{fix} that part of the input. We do have some knowledge about the parts not yet fixed, but it is unreliable. Still, we may be able to assume that our knowledge is rarely wrong.

\subsection*{Structure of the paper}

In the next sections, we study $D_k(P,n)$ for different problems $P$. We consider search theory and group testing in Section 2, sorting in Section 3, finding the largest and the smallest element in a set with an unknown order in Section 4, and the property that a graph is connected in Section 5. We finish the paper with concluding remarks, where we describe a couple variants.

\section{Search theory and group testing}\label{searchgt}

In search theory, $f=id$ is the identity function, i.e. we need to identify an unknown element $x$ from an underlying set of size $n$ (the set of possible inputs). Usually $x$ is called the \textit{defective} element, that causes failures in the whole system and our goal is to locate the defective part. The allowed queries are all the possible YES/NO questions regarding the input. More precisely, we can identify every query with the set of elements where the answer is YES, thus we ask subsets of the underlying set. In particular, the Questioner can ask the set $\{x\}$, showing $D_0(P,n)=1$. 

On the other hand, it is well-known that $D(P,n)=\lceil \log n\rceil$. 
Thus Proposition \ref{altalanos} and Proposition \ref{also} (with $j_1=\dots=j_{k+1}=0$) imply $D_k(P,n)=\min\{k+1,\lceil \log n\rceil\}$.


A natural generalization is to consider the case there are more than one defectives. This version is usually called group testing. For more on this area, see \cite{dw}. In this case the input is a set of elements. Usually we assume that there are exactly $d$ or at most $d$ defectives and the goal is to identify all of them. Let us denote these problems by $P_d$ and $P_{\le d}$. It is well-known that we have $\Omega(d\log (n/ d))=D(P_d,n)\le D(P_{\le d},n)=O(d\log n)$ (see for example \cite{dw}).

Observe first that $D_0(P_d,n)=1$ by asking the complement of the set of defectives, but $D_0(P_{\le d},n)=d$. Indeed, it is enough to ask the sets $\{a\}$ for all the defective elements $a$. On the other hand, for every defective element $a$ we need to ask a set that contains only $a$ among the defective elements, otherwise $a$ could be non-defective. 
For general $k$ we have $D_k(P_d,n)=\min \{k+1,D(P_d,n)\}$ using Proposition \ref{altalanos} and Proposition \ref{also} with $j_1=\dots=j_{k+1}=0$.

The situation, however, is much more complicated for $P_{\le d}$. The same way, using Propositions \ref{altalanos} and \ref{also}, we obtain for $D_k(P_{\le d},n)$ the lower bound $\min \{k+1,D(P_{\le d},n)\}$ and the upper bound $\min \{D(P_{\le d},n)\}$. 
If $n$ is small, this difference disappears, as $D(P_{\le d})$ gives the minimum in both bounds. If $n$ is large we can also determine the exact value of $D_k(P_{\le d},n)$.

\begin{proposition} Let $n> (d-1)2^k$. Then $D_k(P_{\le d},n)=k+d$.

\end{proposition}

\begin{proof} For the upper bound, observe Questioner can always ask singletons that are defective at that point. At every point there are at most $d$ YES and at most $k$ NO answers to such queries, thus the algorithm ends after at most $k+d$ queries.

For the lower bound, consider the following simple algorithm of the Adversary. For the first query $A$, he answers YES if and only if $|A|\ge n/2$. Moreover, he gives the additional information that all the defective elements are in $A$ in this case. Thus after the answer all we know is a set of size at least $n/2$ that contains all the defectives. He answers the same way for the first $k$ queries with respect to the number of remaining elements. After that there is a set of size at least $n/2^k>d-1$ that contains all the defectives. Even if the Adversary used up all the possible changes of the input, the Questioner needs at least $d$ more queries to identify all the defectives.

\end{proof}

\section{Sorting}

In the case of sorting, the input is a set of $n$ different numbers $a_1,\dots,a_n$, and a query corresponds to two numbers $a_i$ and $a_j$ with answer YES if and only if $a_i<a_j$. The goal is to sort the elements, i.e. find the increasing order of them. One can look at it as a special version of the search theory problem studied in Section \ref{searchgt}, as we are identifying one of many possible orderings. However, instead of every possible YES/NO question, we can only ask some special ones. For more on sorting, see for example \cite{knuth}.

Let $S$ denote the sorting problem, then it is well-known that $D(S,n)=\Theta(n\log n)$. On the other hand, $D_0(S,n)=n-1$. Indeed, one can compare the largest and the second largest element, then the second largest and the third largest, and so on. On the other hand,
the queried pairs as edges form a graph, and if that graph is unconnected, we cannot know how their vertices relate to each other.

An important special property of this problem is that those $n-1$ edges are always needed to solve the problem, as if a pair of numbers that are adjacent in the sorted order is not asked as a query, we have no way to show which one is larger. The other $\binom{n-1}{2}$ possible queries are ultimately useless. Hence the Adversary's goal with the changes is to make sure that not too many of these pairs of adjacent numbers are among the queries already asked. 

\begin{proposition}\label{sorti}
$D_1(S,n)=\lceil 3n/2\rceil -2$.
\end{proposition}

\begin{proof} Let us start with the strategy of the Adversary. He answers the first $\lceil n/2\rceil -1$ queries without change. Then he changes the input such a way that none of the queried pairs have that the two vertices are adjacent in this new ordering, thus $n-1$ further queries are needed.

We need to show that the elements can be reordered in the above described way. The queried edges form a graph, with a partial ordering on every connected component. The Adversary extends the partial ordering to an arbitrary total ordering on each connected component. Let $a_1<\dots<a_k$ be a largest component, and consider the sets of components with less than $k$ elements. Pick one with the largest number of elements $\ell\le k$. The Adversary extends the ordering on those components to a total ordering $b_1<\dots<b_\ell$, and let $a_i<b_i<a_{i+1}$ for every $i\le \ell$. For every other component with ordering $c_1<\dots<c_m$, we have $m+\ell\ge k$. The Adversary lets $a_{k-i}<c_i<a_{k-i+1}$ for $i\le m$. Finally, he extends the ordering give above to a total ordering and gives this as the new input.

For $a_i$ and $a_j$ with $i<j$, we have either $b_i$ or $c_{k-i}$ between them in the new ordering. For every pair of vertices from another component, an $a_i$ is between them. This shows that for any of the first $\lceil n/2\rceil -1$ queries, the two elements are not adjacent in the new input, thus $n-1$ further queries are needed.


Let us continue with the strategy of the Questioner. He tries to follow his strategy for the non-deterministic version of the problem, described above: in query $i$ he compares the $i$th and $i+1$st smallest element, until Adversary executes the change. Assume it happens after the $i$th answer, then we have a total ordering of $i+1$ elements $b_1<\dots<b_{i+1}$. After the change, for each other element we know its final place in the ordering, in particular we know how many elements are between $b_j$ and $b_{j+1}$. If there are $p$ elements between $b_j$ and $b_{j+1}$, then $p+1$ queries are enough to prove that their final places are what we suspect, and if $p=0$, then 0 queries are enough. If there are $p$ elements smaller than $b_1$ (or larger than $b_{i+1}$), then $p$ queries are enough for the same. 

This shows that we can find the place of the remaining $n-i-1$ elements with one query for each plus some extra queries: for every $j<i$ we need one extra query for those elements between $b_j$ and $b_{j+1}$ in the final ordering, unless there is no element between $b_j$ and $b_{j+1}$. On the one hand, at most $i-1$ extra queries are needed (thus $n-2$ altogether), because there are $i-1$ intervals where we might need extra queries. On the other hand, at most $n-i-1$ extra queries are needed, as we need at least one element between $b_j$ and $b_{j+1}$ for an extra query. Thus $D_1(S,n)\le \min\{i+n-2,i+2(n-i-1)\}$, which easily implies the statement.
\end{proof}



Let us briefly discuss $D_k(S,n)$. A reasonable algorithm for the Questioner is the straightforward generalization of the above. He looks at the current input $a_1<\dots<a_n$ and asks $a_i,a_{i+1}$ for the smallest $i$ this query has not been asked. 

Against this strategy the Adversary can do the following. He changes after $\ell_1$ queries, and places the elements not queried yet in a balanced way between the elements appearing in earlier queries. Then the $i$th change is applied after $\ell_i$ queries, and the last change is at the last point where he can change the input so that none of the pairs queried earlier are adjacent in the new input, thus $n$ further queries are needed afterwards. Consider an arbitrary point during the algorithm. Let $b_1<\dots<b_p$ be the elements asked by the Questioner in the current round, and $c_1<\dots,c_q$ be the elements asked earlier.

Then the ordering that is fixed is (with a small error) $b_1\dots<b_p<c_1\dots<c_q$, such that in the current ordering the $b_i$s are the $p$ smallest element, and the $c_i$s are above them, divided in a balanced way. The last change has to come when $p+q$ is about $n/2$, thus the extra queries are won in earlier rounds. However, an exact analysis of this would be rather complicated, because of the errors mentioned earlier. Indeed, at a given change the earlier queries gave the ordering $a_1<\dots<a_r$, and whenever later we ask queries containing elements that are between $a_i$ and $a_{i+1}$ in the current ordering, the largest becomes $b_p$, while $c_1=a_{i+1}$. However, we do not necessarily know if the relation $b_p<c_1$ holds. This gives a small error, that might add up during the rounds.

\section{Finding the maximum and the minimum}

A simpler task than the one considered in the previous section is, having the same input and queries, to look for only some special elements. Finding the largest element takes $n-1$ queries both in the deterministic and the non-deterministic way. Indeed, every other element has to be the smaller in an answer, thus $n-1$ queries are needed, and it can be easily achieved by not asking any element that has no chance to be the largest, since it was proved to be smaller than another element earlier. This implies that with any number of changes, $n-1$ queries are needed.

Obviously, the same holds finding the smallest element. Let $L$ denote the problem of finding the largest and the smallest element at the same time. Pohl \cite{pohl} proved that $D(L,n)=\lceil 3n/2\rceil -2$. It is obvious that $D_0(L,n)=n-1$. The lower bound comes from the fact that it is obviously not easier than finding only the largest element, while the upper bound is shown by the queries of the form $a_i,a_{i+1}$, where $a_i$ is the $i$th largest element. Observe that these queries are the only certificate of size $n-1$. Indeed, every element but $a_1$ should be the smaller, and every element but $a_n$ should be the larger in a query. These are $2n-2$ fixed positions, thus no element can be larger (or smaller) in more than one query if we have only $n-1$ queries. Obviously $a_2$ is only smaller than $a_1$, thus the query $a_1,a_2$ has to be present. Then $a_3,a_1$ cannot be asked, thus $a_3,a_2$ can be the only query where $a_3$ is the smaller, hence it has to be asked, and so on.

\begin{theorem}
$D_k(L,n)= \min\{n+k-1,\lceil 3n/2\rceil -2\}$.
\end{theorem}

\begin{proof} First we prove $D_k(L,n)\le n+k-1$. Observe that this, together with $D(L,n)=\lceil 3n/2\rceil -2$ proves $D_k(L,n)\le \min\{n+k-1,\lceil 3n/2\rceil -2\}$. We use induction on $k$, the base case $k=0$ was described above. In the first round, the Questioner asks $a_i,a_{i+1}$ as the $i$th query, where $a_i$ is the $i$th largest element. Assume there is a change after the $j$th query. Then we know $a_2,\dots,a_j$ cannot be the largest, nor the smallest elements. Thus we need to find the largest and smallest element among the other ones, hence $D_{k-1}(L,n-j+1)\le n-j+k-1$ further queries are enough. Altogether there were at most $n+k-1$ queries, finishing the proof of the upper bound..

We continue with the Adversary's strategy. We say that an element is \textit{out} if it was both smaller and larger in queries.
During the first part of the algorithm, the Adversary answers such a way that no element goes out, whenever he can. In this case he also changes the input the following way. Assume $A$ is the set of elements that are out, $B$ is the set of elements that are not out and were smaller in a query, and $C$ is the set of elements that are not out and were larger in a query (and there is a set $D$ of $n-|A|-|B|-|C|$ elements that have not appeared in a query). Then in the new input the elements of $B$ are the smallest (ordered arbitrarily), then come the elements of $A$, then the elements of $D$ (ordered arbitrarily), and finally the elements of $C$ (ordered arbitrarily). Note that any such ordering is consistent with the answers so far, except for the ordering inside $A$, where the Adversary chooses an arbitrary ordering consistent with the answers.

If during this part the Adversary has to answer such a way that an element is out (because two elements of $B$ or two elements of $C$ are queried), then he answers without changing the input. This part ends after $\min\{k,\lceil n/2\rceil\}$ answers where no element goes out (thus after at most $k$ changes). In the second part, he does not change the input anymore. 

Observe that the Adversary does not change the input with his very first answer, and no element goes out after the first answer.
Thus there are at least $1+\min\{k,\lceil n/2\rceil\}+|A|$ queries in the first part. In the second part, all but one element of $B$ has to be larger in a query, and all but one element of $C$ has to be smaller in a query. These all require distinct queries, as the elements of $B$ are smaller than the elements of $C$ in the current input. Thus we have at least $1+\min\{k,\lceil n/2\rceil\}+|A|+|B|-1+|C|-1=\min\{n+k-1,\lceil 3n/2\rceil -2\}$ queries.
\end{proof}

\section{Connected graphs}

The input is an unknown graph $G$ on $n$ vertices, 
the allowed queries correspond to pairs of vertices $u,v$, and the answer tells if $uv$ is an edge of $G$ or not. Let $C$ denote the problem of finding out if $G$ is connected or not.

It is well-known (see for example \cite{ly}) that $D(C,n)=\binom{n}{2}$, i.e. the connectedness is an \textit{evasive} graph property. In other words, Questioner has to ask all the pairs and completely identify the graph $G$ in order to find out if it is connected. A simple strategy of the Adversary that shows this is the following. He answers NO to every query unless it would make the graph unconnected, i.e. he changes only that one edge of the input. This shows $D_{n-2}(C,n)=\binom{n}{2}$, as there are always $n-1$ YES answers in this strategy, and the last one arrives after the last query (thus the Adversary can answer NO instead of YES for that query, avoiding the last change).

On the other hand, $D_0(C,n)=\lfloor n^2/4\rfloor$. Indeed, the Adversary can choose a graph $G$ consisting of two components of size $\lfloor n/2\rfloor$ and $\lceil n/2\rceil$. Then $G$ is unconnected, but to show this, one has to ask all the possible edges between the two components.

Let $T(n,k)$ denote the Tur\'an graph, i.e. the $k$-partite complete balanced graph on $n$ vertices, and $t(n,k)$ denote the number of its edges. Here balanced means every part has size $\lfloor n/k\rfloor$ or $\lceil n/k\rceil$. Note that this graph has to most number of edges among $k$-partite graphs.

\begin{theorem}\label{connec}
$t(n,k+2)\le D_k(C,n)\le t(n,k+2)+n-1$.
\end{theorem}

\begin{proof}
Let us start with the Adversary's strategy. He chooses as the original input the complement $G$ of the Tur\'an graph $T(n,k+2)$. Let $A_1,\dots,A_{k+2}$ be the parts in the Tur\'an graph; they are cliques in $G$. He always answers that vertices in the same $A_i$ are connected; in fact he gives this information for free, and in the followings we consider those edges to be known. The Adversary changes the input whenever the NO answer would make the graph unconnected (and would finish the algorithm). In that situation he changes the current non-edge to an edge. This way the graph will be not connected at the end, but the last change makes the graph have only two components. Assume an edge $uv$ of the Tur\'an graph has not been asked. Then at the end of the algorithm $u$ and $v$ cannot be in the same component. Indeed, in that case there was a query $u'v'$ at one point that connected two components, one of them having $u$ and the other having $v$. But then the answer would have been no to that query, as it does not make the graph unconnected. Thus $u$ and $v$ are in different components, but then the graph could be connected, a contradiction.

For the upper bound, we apply induction on $k$, the base step $k=0$ is described before the theorem is stated. Note that if $k\ge n-2$, the upper bound is larger than $\binom{n}{2}$, thus the statement is trivial. Hence from now on we assume $k\le n-3$.

Let us describe the Questioner's strategy. First, he asks the edges of a spanning forest consisting of spanning trees in each component. Whenever the current graph does not have such a spanning forest of edges already asked, he picks the smallest number of edges that would form a spanning forest with some of the edges we know are in the graph. Whenever the current graph has such a spanning forest, he picks the smallest cut and asks its edges. 

\begin{claim}\label{clma} After finding a spanning forest (of the graph $G'$ that is the input graph at that point), if there are $i$ changes left, at most $t(n,i+2)$ queries are needed.

\end{claim}

\begin{proof}[Proof of Claim] Observe first that each of the connected components of $G'$ contains a spanning tree. If there are at most $i+2$ connected components of $G'$, then (as it is pointless to ask edges with both endpoints in the same component) the number of queries we need to ask cannot be more than the number of edges in the complement graph of $G'$, which is at most $t(n,i+2)$.

If there are more than $i+2$ components of $G'$, let $a_1,\dots,a_{i+1}$ be the size of the components we try to separate from the remaining vertices, and let $a_{i+2}=n-\sum_{j=1}^{i+1} a_j$. Then the number of queries we ask after finding the spanning tree is at most $\prod_{1\le j<l\le i+2}a_ja_l$, which is the number of edges in the complete $(i+2)$-partite graph with parts of size $a_j$, thus at most $t(n,i+2)$.
\end{proof}

Let us return to the proof of the theorem. If there is no change while finding the spanning forest, then the above claim finishes the proof.

If there is a change during the queries corresponding to the very first spanning forest, then we continue till we find a spanning forest, by asking edges that are present in the current input graph, but are not inside a connected component of the known edges. Observe that after at most $n-1$ YES answers, we have a spanning forest in the current graph. If we had $k-i$ NO answers during that, that means $k-i$ changes, thus at most $t(n,i+2)$ queries are needed afterwards, by Claim \ref{clma}. This means there are $n-1+k-i+t(n,i+2)$ queries altogether. Observe that $t(n,i)<t(n,i+1)$, whenever $i<n$. This implies $k-i+t(n,i+2)\le t(n,k+2)$, whenever $k\le n-2$, finishing the proof. 
\end{proof}

\section{Concluding remarks}

We have defined a quantity that is between the deterministic and non-deterministic query complexity. We studied it for a couple examples, but there are countless many other, equally interesting questions. Also, we were unable to completely determine $D_k(P,n)$ even for most of the problems we considered.

Here we describe some potential variants of the model we have studied.

In the graph case, or more generally if the input is disjoint union of the possible queries (for example in case of any kind of Boolean functions), one could restrict the change to the query being asked. Note that in fact this is how the Adversary changes the input graph in Theorem \ref{connec} anyway. For this model, we can consider the following motivation: 
Questioner has some outdated information, for example a map, and some things have changed since that was published. However, it is a reasonable assumption that there was a limited number of changes. Also the changes are independent, i.e. when we learn our information was incorrect at one place, we have no reason to assume any particular other change (compare this to the bigger paradigm shifts in one of the motivations of the model studied in this paper).

The \emph{changes} in our study are in some sense opposite to \textit{lies} in other studies. There the Adversary has a different option: he can give arbitrary false answers at most $l$ times. This is in addition to his ability to change the input as many times as he wants it. Still, it would be interesting to combine these abilities and study an Adversary who can change the input $k$ times and lie $l$ times. In this case we would not assume the Questioner knows everything.

There are other models of lies, when a fixed proportion of the answers can be false, or when every answer is false with a probability $p$, independently from each other. Similarly, we could modify our model and allow changes proportional to the number of queries, or define a distribution on the possible input graphs at every point, and choose one randomly.

We could also combine changes and lies in a more complicated way. For example, the Adversary could tell that an earlier answer of his was a lie, because he changed the input at that point.

\bigskip
\textbf{Funding}: Research supported by the National Research, Development and Innovation Office - NKFIH under the grants SNN 129364, KH 130371 and K 116769 and by the J\'anos Bolyai Research Fellowship of the Hungarian Academy of Sciences.

\end{document}